\documentclass[letterpaper]{article} 
\usepackage{aaai2026}  
\usepackage{times}  
\usepackage{helvet}  
\usepackage{courier}  
\usepackage[hyphens]{url}  
\usepackage{graphicx} 
\usepackage{natbib}  
\usepackage{caption} 
\usepackage{algorithm}
\usepackage{algorithmic}

\usepackage{newfloat}
\usepackage{listings}
\DeclareCaptionStyle{ruled}{labelfont=normalfont,labelsep=colon,strut=off} 
\floatstyle{ruled}
\newfloat{listing}{tb}{lst}{}
\floatname{listing}{Listing}
\usepackage{multirow}
\usepackage{amsmath}
\usepackage{makecell}
\usepackage{subfigure}
\usepackage{enumitem}
\usepackage{amssymb}

\usepackage{amsthm}
\usepackage{booktabs}
\newcommand{\method}{\texttt{FedAU$^2$}}
\newtheorem{theorem}{Theorem}
\newtheorem{corollary}{Corollary}

\title{FedAU2: Attribute Unlearning for User-Level Federated Recommender Systems with Adaptive and Robust Adversarial Training}
\author{
Yuyuan Li\textsuperscript{\rm 1},
Junjie Fang\textsuperscript{\rm 1},
Fengyuan Yu\textsuperscript{\rm 2},
Xichun Sheng\textsuperscript{\rm 3},
Tianyu Du\textsuperscript{\rm 2},
Xuyang Teng\textsuperscript{\rm 1},
Shaowei Jiang\textsuperscript{\rm 1},
Linbo Jiang\textsuperscript{\rm 4},
Jianan Lin\textsuperscript{\rm 4},
Chaochao Chen\textsuperscript{\rm 2}\thanks{Corresponding author: zjuccc@zju.edu.cn}
}
\affiliations{
\textsuperscript{\rm 1}Hangzhou Dianzi University\\
\textsuperscript{\rm 2}Zhejiang University\\
\textsuperscript{\rm 3}Macao Polytechnic University\\
\textsuperscript{\rm 4}Ant Group\\
}

\begin{document}

\maketitle

\begin{abstract}
    Federated Recommender Systems (FedRecs) leverage federated learning to protect user privacy by retaining data locally. 
    However, user embeddings in FedRecs often encode sensitive attribute information, rendering them vulnerable to attribute inference attacks. 
    Attribute unlearning has emerged as a promising approach to mitigate this issue. 
    In this paper, we focus on user-level FedRecs, which is a more practical yet challenging setting compared to group-level FedRecs.
    Adversarial training emerges as the most feasible approach within this context.
    We identify two key challenges in implementing adversarial training-based attribute unlearning for user-level FedRecs:
    i) mitigating training instability caused by user data heterogeneity, and ii) preventing attribute information leakage through gradients.
    To address these challenges, we propose FedAU2, an attribute unlearning method for user-level FedRecs.
    For \textbf{CH1}, we propose an adaptive adversarial training strategy, where the training dynamics are adjusted in response to local optimization behavior.
    For \textbf{CH2}, we propose a dual-stochastic variational autoencoder to perturb the adversarial model, effectively preventing gradient-based information leakage.
    Extensive experiments on three real-world datasets demonstrate that our proposed FedAU2 achieves superior performance in unlearning effectiveness and recommendation performance compared to existing baselines.
    \end{abstract}

\section{Introduction}\label{sec:intro}
Recommender Systems (RSs) are integral to modern online platforms, delivering personalized recommendations based on user preferences~\cite{hasan2024based, lu2025dmmd4sr,pbat2023su,feng2024fine}.
Traditional RSs are typically designed in centralized settings, requiring users to provide all raw interaction data, such as clicks and purchases. 
However, this centralized approach raises significant concerns regarding user privacy and potential data misuse~\cite{qu2024towards, zhou2023fastpillars}.
To address these concerns, Federated Learning (FL)~\cite{mcmahan2017communication,liu2024fedbcgd, liu2024qoe,zhang2025fedfact,chen2024integration} has emerged as a privacy-preserving paradigm for training recommendation models, which are known as Federated RSs (FedRecs)~\cite{sun2022survey}.
In FL, users' raw data remains on their local devices, and a global model is collaboratively trained through local updates and aggregation.

Recent data privacy regulations, such as the GDPR~\cite{protection2018general} and CCPA~\cite{illman2019california}, emphasize the \textit{right to be forgotten}, granting users the ability to withdraw personal data, including their influence on models and sensitive information~\cite{li2024survey,feng2025controllable}.
RSs represent critical application scenarios that heavily rely on personal data and often embed users' attribute information, rendering them susceptible to attribute inference attacks~\cite{chen2024post}.
Attribute unlearning has emerged as a promising approach to eliminate users' attribute information from models~\cite{feng2025plug, yu2025lego}.
However, while FedRecs preserve the privacy of local data, it does not provide mechanisms to remove users' attribute information~\cite{hu2024user}, thereby failing to satisfy the right to be forgotten. 
Therefore, enabling attribute unlearning in FedRecs is imperative to uphold users' privacy.

FedRecs can be classified into two categories~\cite{sun2022survey}: user-level and group-level.
In user-level FedRecs, each client corresponds to an individual user (e.g., a personal smartphone), while in group-level FedRecs, each client aggregates data from multiple users (e.g., within an organization).
Attribute unlearning methods (e.g., distribution alignment) developed for group-level settings~\cite{hu2024user, wu2025aegis, lu2025dammfnd, feng2025survey} typically aggregate data from multiple users, rendering them inapplicable to user-level FedRecs, where data remains strictly isolated on individual devices.

In this paper, we focus on attribute unlearning in user-level FedRecs, which is more generalized and more challenging~\cite{sun2022survey}. 
As distribution alignment~\cite{wu2025aegis} is inapplicable in this setting, adversarial training~\cite{feng2025raid, wang2025robin} emerges as the most feasible approach. Nevertheless, we identify two key challenges in its implementation:
\textbf{CH1}: How to make adversarial training stable in federated settings, i.e., the cold start problem~\cite{zhang2024bidirectional}.
Specifically, traditional adversarial learning methods require pre-training to stabilize the training process~\cite{ganhor2022unlearning}
However, in FedRecs, due to the heterogeneity of user data and the randomness of sampling~\cite{zhang2024federated}, pre-training may lead to unstable training processes.
\textbf{CH2}: How to prevent attribute information leakage during unlearning in FedRecs.
Recent studies have demonstrated that the model gradient transmitted during server-client communication can expose users' raw data through reconstruction attacks~\cite{zhao2020idlg, zhu2019deep},
with attribute labels being particularly vulnerable to such leakage.

To address these challenges, we propose \method{}, an \textbf{A}ttribute \textbf{U}nlearning method for \textbf{U}ser-level \textbf{Fed}Recs.
Specifically, we perform adversarial learning locally on each client and aggregate model updates on the server, thereby eliminating the reliance on group-level attribute information.
For \textbf{CH1}, 
we propose a decentralized and adaptive training strategy that operates at the user level to avoid the global pre-training. 
Specifically, each user autonomously adjusts their adversarial training in response to their own adversarial prediction outcome, enabling fine-grained control and enhanced adaptation.
For \textbf{CH2}, we propose a Dual-Stochastic Variational AutoEncoder (DSVAE), which is integrated into the adversarial model.
Our theoretical analysis reveals that the dual stochastic design helps perturb the adversarial model to enhance robustness against gradient-based reconstruction attacks.
We summarize the main contributions of this paper as follows:
\begin{itemize}[leftmargin=*]\setlength{\itemsep}{-\itemsep}
    \item We propose \method{}, a novel attribute unlearning method for user-level FedRecs.
    We identify two key challenges: i) avoiding adversarial training instability in federated settings, and ii) preventing gradient-based attribute leakage.
    \item For \textbf{CH1}, we propose an adaptive adversarial training strategy that dynamically adjusts each user's training process based on their local updates. This design
    enhances both the stability and efficiency of adversarial optimization in an adaptive manner.
    \item For \textbf{CH2}, we integrate a DSVAE into the adversarial model. Theoretical analysis reveals that our proposed DSVAE perturbs the parameters to effectively prevent attribution information leakion in FedRecs.
    \item We conduct extensive experiments on three real-world datasets across representative recommendation models. The results demonstrate that our proposed \method{} significantly outperforms existing baselines.
\end{itemize}

\section{Related Work}\label{sec:relate}
\subsection{Federated Recommendation Systems}
FedRecs leverage federated learning to enable collaborative model training without sharing raw user data in RSs.
Based on the granularity of user participation~\cite{sun2022survey}, FedRecs can be categorized into user-level (where each device acts as a client) and group-level (where clients represent user groups or a data silo) settings.
Most FedRecs are developed based on the user-level setting~\cite{chai2020secure, perifanis2022federated}
Latest user-level FedRecs also explore personalized methods 
to better capture user-specific preferences under heterogeneous data~\cite{li2023federated}.

The key distinction between these two settings lies in data and client heterogeneity~\cite{sun2022survey}. 
Compared to group-level settings, user-level FedRecs face extremely sparse and non-IID data.
This introduces substantial challenges for unlearning methods that depend on cross-user information~\cite{hu2024user, wu2025aegis, ENCODER, FineCIR}.
In this work, we focus on the user-level setting, which remains the most prevalent and technically challenging scenario for federated recommendation.

\subsection{Recommendation Attribute Unlearning}
RSs can implicitly encode sensitive user attributes (e.g., gender and age) into user embeddings, making them vulnerable to attribute inference attacks~\cite{feng2025raid, liu2025fine}.
Attribute unlearning~\cite{li2023making} has emerged 
to mitigate this issue.
In centralized RSs, researchers utilize adversarial learning~\cite{ganhor2022unlearning} and post-training tuning~\cite{li2023making} to achieve attribute unlearning.

FedRecs also face risks of attribute leakage because federated learning cannot prevent attribute inference~\cite{hu2024user}.
Traditional collaborative filtering methods
and latest personalized FedRecs
both encode user attributes into embeddings, increasing privacy risks.
To mitigate such risks, Hu et al. propose a user-consented approach
for attribute unlearning~\cite{hu2024user}.
Aegis~\cite{wu2025aegis} extends post-training methods to the federated setting.
However, these methods 
rely on group-level user data, which is infeasible in user-level FedRecs.
For user-level FedRecs, Zhang et al. propose a
Differential Privacy (DP) framework~\cite{zhang2023comprehensive}.
However, DP-based methods typically lack attribute specificity (i.e., cannot selectively unlearn a particular attribute) and often significantly degrade recommendation performance. 
These gaps highlight the need for an attribute unlearning approach for user-level FedRecs that can effectively mitigate attribute leakage while preserving recommendation performance.

\begin{figure*}[t]
    \centering
    \includegraphics[width=\linewidth]{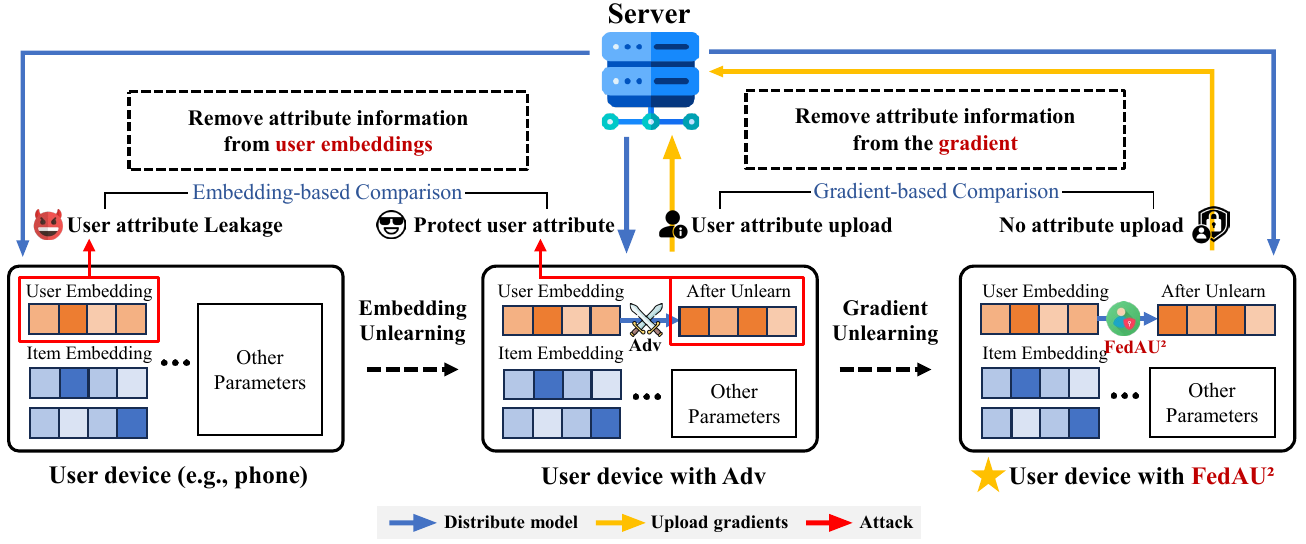}
    \caption{
        Comparison of three user device types in federated recommendation. 
        i) on the left, standard clients expose attribute information in embeddings; 
        ii) in the middle, adversarial clients can prevent embedding leakage but still suffer from gradient leakage; 
        and iii) on the right, our proposed \method{} eliminates information leakage from both embeddings and gradients.
    }    
    \label{fig:fedau2_story}
\end{figure*}

\subsection{Gradient-based Reconstruction Attacks}
Gradient-based reconstruction attacks (e.g., DLG) can reconstruct training data by analyzing raw gradients~\cite{zhu2019deep, duada25su}.
Subsequently, iDLG~\cite{zhao2020idlg} leverages analytical derivation to directly obtain real labels, making the attack more efficient and reliable. 

Current defenses 
fall into two categories: feature-side and label-side.
Feature-side methods protect input data~\cite{huang2020instahide, sun2021soteria},
but do not secure output labels, which are crucial in user-level FedRecs. 
In contrast, label-side methods protect label-related privacy at the output layer.
Label-DP noise~\cite{ghazi2024labeldp} and soft labels~\cite{struppek2023careful, he2024labobf} trade privacy for accuracy, and even hinder convergence in user-level FedRecs~\cite{wang2024horizontal, OFFSET, pillarhist}.
None of these methods scales well to user-level FedRecs.

\section{Preliminaries}\label{sec:pre}
\subsection{User-level Federated Recommendation}
Let \( \mathcal{U} \) and \( \mathcal{I} \) denote the sets of users and items, respectively.  
Each user \( u \in \mathcal{U} \) holds a private interaction vector \( \mathbf{x}_u \in \{0,1\}^{|\mathcal{I}|} \), where \( \mathbf{x}_{ui} \in \{0,1\} \) indicates whether user \( u \) has interacted with item \( i \in \mathcal{I} \).
The goal of FedRecs is to learn a scoring function \( f_\Theta : \mathcal{U} \times \mathcal{I} \rightarrow \mathbb{R} \),  
where \( f_\Theta(u, i) \) denotes the predicted score between user \( u \) and item \( i \).
To extract features of users and items, the system typically encodes users and items via \( \mathbf{em}_u = \phi_u(u) \) and \( \mathbf{em}_i = \phi_i(i) \),  
and computes the matching score between them using a scoring function \( s_\psi(\cdot) \). 
The recommendation objective can be formalized as:
\begin{equation}
    \min_{\Theta} \mathcal{L}_{\text{rec}} = \sum_{u \in \mathcal{U}} \sum_{i \in \mathcal{I}} \ell\left( 
     s_\psi(\mathbf{em}_u, \mathbf{em}_i) ,\ \mathbf{x}_{ui} \right),
\end{equation}
where \( \ell(\cdot) \) is a loss function (e.g., BPR), and  
\( \Theta = \phi_u \cup \phi_i \cup \psi\).
In FedRecs, user-side parameters \( \phi_u \) are locally updated and remain private,  
whereas item-side parameters \( \phi_i \) and model parameters \( \psi \) are uploaded and aggregated on the server.

\subsection{Adversarial Training}
Let \( \mathbf{em}_u \in \mathbb{R}^d \) denote the embedding of user \( u \), and let \( y \in \mathcal{Y} \) be a protected attribute (e.g., gender or age).  
We introduce an adversarial network \( h_\omega(\cdot) \), which takes \( \mathbf{em}_u \) as input and attempts to predict the protected attribute \( y \).  
The model learns to preserve recommendation quality while removing sensitive information from \( \mathbf{em}_u \), leading to the following min-max objective:
\begin{equation}
    \min_{\Theta} \max_\omega \mathcal{L}_{\text{rec}} - \lambda \cdot \sum_{u \in \mathcal{U}} \mathcal{L}_{\text{adv}}(\mathbf{em}_u, y),
\end{equation}
where \( \mathcal{L}_{\text{adv}} \) is the adversarial loss, implemented as a classification loss using cross-entropy:
\[
    \mathcal{L}_{\text{adv}}(\mathbf{em}_u, y) = \text{CE}(h_\omega(\mathbf{em}_u), y),
\]

To solve this optimization problem, we adopt a Gradient Reversal Layer (GRL)~\cite{ganin2015unsupervised} between \( \mathbf{em}_u \) and \( h_\omega \).  
GRL preserves the forward pass and reverses the backward gradient.  
Converting the min-max into a minimization.:
\begin{equation}
    \min_{\Theta, \omega} \mathcal{L}_{\text{rec}} + \lambda \cdot \sum_{u \in \mathcal{U}} \text{CE}(h_\omega(\text{GRL}(\mathbf{em}_u)), y).
    \label{eq:adv_training}
\end{equation}

\subsection{Reconstructing Data via DLG}

DLG reconstructs private user data by optimizing dummy inputs and labels to match observed gradients.  
Given the observed gradient \( \nabla W = \nabla_W \mathcal{L}(h_\omega(\mathbf{em}), y) \),  
the attacker initializes dummy variables \( \mathbf{em}', y' \sim \mathcal{N}(0, 1) \),  
and iteratively updates them to minimize the distance between dummy and real gradients:
\begin{equation}
(\mathbf{em}^*, y^*) = \arg\min_{\mathbf{em}', y'} \left\| \nabla W' - \nabla W \right\|^2,
\label{eq:dlg_optimization}
\end{equation}
where \( \nabla W' = \nabla_W \mathcal{L}(h_\omega(\mathbf{em}'), y') \) is the gradient computed from dummy data.  
This optimization exploits the differentiability of the gradient distance w.r.t. \( \mathbf{em}' \) and \( y' \).  
By minimizing this distance, the dummy data \((\mathbf{em}', y')\) gradually approximates the true training sample \((\mathbf{em}, y)\).

\section{Methodology}
\label{sec:methodology}

\subsection{Overview of FedAU2}

To achieve user-level attribute unlearning in FedRecs, we employ the adversarial training approach.
However, the direct application of adversarial training introduces two key challenges (\textbf{CH1}: stable training and \textbf{CH2}: privacy preservation).
To tackle these challenges, we propose \method{}, an adaptive and robust user-level attribute unlearning method for FedRecs.

\method{} demonstrates a significant advantage over both standard FL clients and the conventional adversarial training approach.
As shown in Figure~\ref{fig:fedau2_story}, i) the embeddings in a \textit{standard FL client} may leak sensitive attribute information; 
ii) a client equipped with a \textit{conventional adversarial network} mitigates embedding leakage but remains vulnerable to gradient-based reconstruction attacks; 
and iii) \method{} effectively prevents both embedding and gradient leakage, thereby offering enhanced privacy protection.

Specifically, \method{} integrates two synergistic components into adversarial training to enable adaptive and robust attribute unlearning, 
i.e., Selective Unlearning Trigger (SUT) for stable training 
and DSVAE to preserve gradient privacy. 
Figure~\ref{fig:fedau2_overview} illustrates the 
workflow of \method{}. 

\subsection{Selective Unlearning Trigger}
Adversarial training helps mitigate privacy leakage but often introduces instability, i.e., the cold start problem of adversarial models (\textbf{CH1}).
Moreover,
global pretraining strategies~\cite{ganhor2022unlearning} prove inadequate for FedRecs. 
The heterogeneity of user data and the randomness of client sampling result in uneven user participation, which in turn limits the generalization of the adversarial model in federated learning~\cite{yang2024understanding}.

The instability of adversarial training primarily stems from the trade-off governed by the perturbation budget~\cite{andriushchenko2020understanding}. 
There are two extreme cases in practice:
i) \textit{excessive perturbation budget:} aggressive perturbations obscure task-relevant features, rendering the adversarial loss uninformative and driving training into noisy updates; and
ii) \textit{insufficient perturbation budget:} mild perturbations are unable to erase the sensitive attributes derived from standard training, leaving privacy-related information embedded in the learned representations.

\begin{figure}[t]
    \centering
    \includegraphics[width=\linewidth]{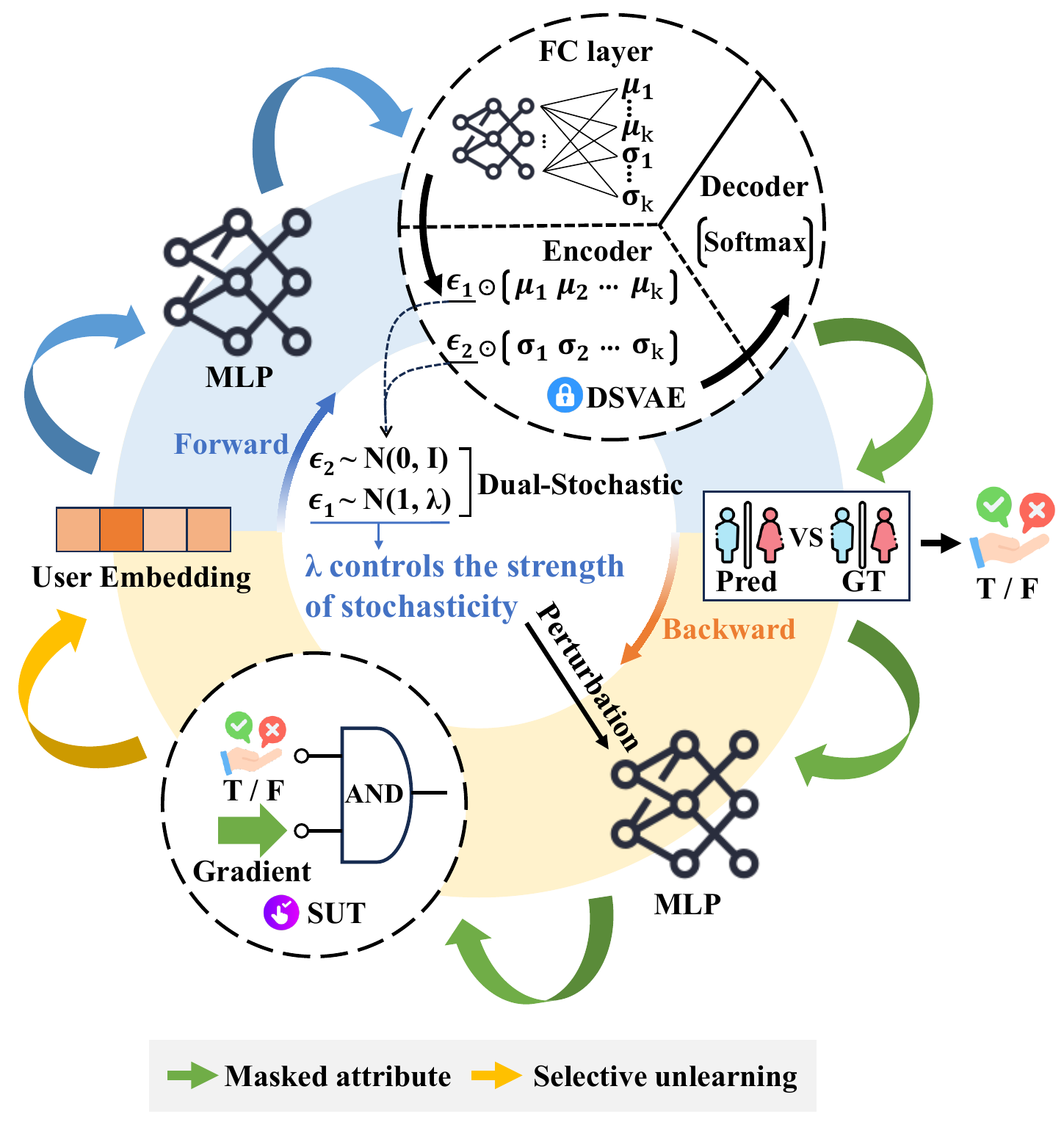}
    \caption{Workflow of \method{}. 
    During the forward pass, DSVAE injects dual-stochasticity
    , effectively masking attribute information embedded in the gradients. 
    During the backward pass, SUT dynamically adjusts the perturbation budget based on the prediction outcomes, enabling stable adversarial training.}
    \label{fig:fedau2_overview}
\end{figure}

To balance the trade-off in perturbation budget, we propose a SUT for user-level FedRecs, which enables efficient and adaptive control of the perturbation budget at the individual user level.  
Motivated by the
adaptive perturbation control at the sample level~\cite{balaji2019instance}, we incorporate this adaptive design into our adversarial training objective for user-level FedRecs.  
Specifically, the per-user training objective is formulated as:
\begin{equation}
    \min_{\theta, \omega} \; \mathcal{L}^{(u)}_{\text{rec}} 
    + \lambda \cdot \mathcal{L}_{\text{CE}}(h_\omega(\text{GRL}(\mathbf{em}_u; \epsilon_u)), y_u),
    \label{eq:tradeoff-abs}
\end{equation}
where \( \epsilon_u \) determines the strength of adversarial unlearning and is dynamically adjusted based on the estimated distance to the decision boundary:
\begin{equation}
    \epsilon_u = \frac{\tau}{\left\| \nabla_{\mathbf{em}_u} \mathcal{L}(h_\omega(\mathbf{em}_u), y_u) \right\|_2}, 
    \label{eq:dynamic-budget}
\end{equation}
where \( \tau \) is a global scaling factor. 
\( \epsilon_u \) serves as a local margin estimator between the user embedding \( \mathbf{em}_u \) and the decision boundary of the adversarial classifier.  

Directly applying Eq.~(\ref{eq:dynamic-budget}) on the client side introduces significant computational overhead, as it requires computing the gradient at every training iteration.  
Existing studies \cite{elsayed2018large} show  that correct predictions generally correspond to smaller gradient norm, 
suggesting that the embedding \( \mathbf{em}_u \) resides near a local optimum and is far from the decision boundary.
Conversely, misclassified samples tend to lie closer to the boundary, where the gradient norm are higher, indicating higher sensitivity to perturbations, but less informative for unlearning.
Based on this observation, we introduce a simplified binary variant of the perturbation budget SUT defined as: \(\epsilon_u = \epsilon\) if \(\hat{y}_u = y_u\), and \(\epsilon_u = 0\) otherwise,
where \( \hat{y}_u = \arg\max h_\omega(\mathbf{em}_u) \) is the adversarial prediction, and \( \epsilon \) is a constant controlling the gradient reversal strength in the GRL layer.

The binary design of \( \epsilon_u \) greatly simplifies per-user computation while offering clear interpretability: 
if the adversarial classifier misclassifies the label (insufficient discriminatory power), adversarial training is skipped; 
if the prediction is correct (captured label-related information), unlearning is applied accordingly.
SUT effectively balances the trade-off in perturbation budget through user-level adaptive adjustment, eliminating the need for global pretraining.
This design enables more stable adversarial training in federated settings.

\subsection{Dual-Stochastic Variational Autoencoder}
\label{subsec:dsvae}

For \textbf{CH2}, we observe that even if the attribute information can be unlearned from embeddings, the adversarial gradients can still leak sensitive labels~\cite{zhu2019deep}.
To address this challenge, we propose DSVAE, which injects stochasticity to mask sensitive labels in the gradients. 
We motivate its design by comparing how three different model architectures (original, VAE, and DSVAE) behave under DLG-based attacks, starting with a unified theoretical analysis of how DLG reconstructs labels, which generalizes across architectures.

DLG aims to reconstruct the true user label \( y \) by optimizing a dummy label \( y^* \) that closely approximates it.
Since the label \( y \) is inferred 
through the final layers of the model~\cite{zhao2020idlg},  
our analysis specifically focuses on the final layer preceding the softmax output.
Given a class \( i \), we define \( \hat{y}_i' \) as its predicted probability during the reconstruction attack, 
\( \nabla W_i \) as the corresponding observed gradient vector from the final layer, and \( y_i^* \) as optimal dummy label.

%
\begin{theorem} 
\label{theorem:closed_yi}
Given a class \( i \), let \( \hat{y}_i' \) denote the predicted probability during the reconstruction attack,  
\( \nabla W_i \) the observed gradient of the final layer, and \( y_i^* \) the optimal dummy label.  
Then, \( y_i^* \) admits the following closed-form solution:

\begin{equation}
    y_i^* = \hat{y}_i' - \delta_i,
    \label{eq:closed_form_yi_dummy}
\end{equation}
where \( \delta_i = \mathcal{G}(\nabla W_i) \), a function with of the gradient \( \nabla W_i \).
\end{theorem}

\begin{proof}
    The proof can be found in Appendix B.1.
\end{proof}

As shown in Eq.~(\ref{eq:closed_form_yi_dummy}), \( y_i^* \) is depended on two part:
i) \( \hat{y}_i' \)
approximates the client-side prediction \( \hat{y}_i \) from the reconstructed embedding \( \mathbf{em}' \); and 
ii) a correction term \( \delta_i \) that reflects the preference encoded by the gradient \( \nabla W_i \) for label \( i \).
%
By leveraging both \( \hat{y}_i' \) and \( \delta_i \),  
DLG can infer the user's attribute.
Thus, the key to preventing DLG from reconstructing labels lies in perturbing these two components.  
In the following, we analyze the behavior of \( \hat{y}_i' \) and \( \delta_i \) under different model architectures.  

\subsubsection{Original MLP}
Following prior work~\cite{ganhor2022unlearning}, we use an MLP as the adversarial classifier.
In the original MLP, the linear transformation is deterministic without any form of stochasticity.  
As a result, DLG can easily reconstruct the user label information by exploiting both \( \hat{y}_i' \) and \( \delta_i \).

\begin{corollary}
    \label{corollary:closed_form_ori}
    Let \( \mathbf{h}' \) denote the input to the final linear layer of the MLP during the reconstruction attack.
    Then, the correction term \( \delta_i \) is given by:
    \begin{equation}
        \delta_i = \frac{(\mathbf{h}')^\top \nabla W_i}{\|\mathbf{h}'\|^2}.
        \label{eq:delta_i_ori}
    \end{equation}
\end{corollary}

\begin{proof}
The proof can be found in Appendix B.2.
\end{proof}

For \( \hat{y}_i' \), DLG can accurately reconstruct the user embedding from the original MLP~\cite{zhu2019deep}, effectively approximating the client-side prediction \( \hat{y}_i \) using the reconstructed embedding \( \mathbf{em}' \).  
For \( \delta_i \), DLG can accurately capture the preference encoded in the gradient \( \nabla W_i \) through Eq.~(\ref{eq:delta_i_ori}),  
where the sign of \( \nabla W_i \) differs when \( i \) corresponds to the true label~\cite{zhao2020idlg}, directly revealing the user label \( y_i \).

\subsubsection{MLP with VAE}
While the Variational Autoencoder (VAE) architecture injects stochasticity during training, this primarily affects \( \hat{y}_i' \),  
leaving the correction term \( \delta_i \) intact and still exploitable by DLG for reconstructing user labels.

VAE assumes that each input is associated with a latent variable \( \mathbf{u} \sim \mathcal{N}(\mathbf{0}, \mathbf{I}) \), which is learned by maximizing the evidence lower bound.
To enable backpropagation, it applies the reparameterization trick:
\[
\mathbf{u} = \boldsymbol{\mu} + \boldsymbol{\sigma} \odot \boldsymbol{\epsilon}, \quad \boldsymbol{\epsilon} \sim \mathcal{N}(\mathbf{0}, \mathbf{I}),
\]
where \( \boldsymbol{\mu} \) and \( \boldsymbol{\sigma} \) are 
tuned to capture the latent distribution.
Note that VAE injects stochasticity solely via the \( \boldsymbol{\sigma} \) path through the sampling of \( \boldsymbol{\epsilon} \),  
whereas the \( \boldsymbol{\mu} \) path remains unaffected.  
As a result, DLG can still reconstruct user labels by leveraging the deterministic \( \boldsymbol{\mu} \) path.

\begin{corollary}
    \label{corollary:closed_form_vae}
    Let \( \mathbf{z}' \) denote the input to the VAE during the reconstruction attack,
    and \( \nabla W_i^\mu \) denote the observed gradient vector associated with the \( i \)-th logit along the \( \boldsymbol{\mu} \) path.
    Then, the correction term \( \delta_i \) derived from \( \nabla W_i^\mu \) is given by:
    \begin{equation}
        \delta_i = \frac{(\mathbf{z}')^\top \nabla W_i^\mu}{\|\mathbf{z}'\|^2}
        \label{eq:closed_form_vae}
    \end{equation}
\end{corollary}

\begin{proof}
    The proof can be found in Appendix B.3.
\end{proof}
For \( \hat{y}_i' \), VAE hinders accurate reconstruction of user embeddings through stochastic perturbations~\cite{scheliga2022precode}, thereby disrupting the direct mapping from \( \hat{y}_i' \) to the client-side prediction \( \hat{y}_i \).  
In contrast, for \( \delta_i \), since VAE does not introduce stochasticity along the \( \boldsymbol{\mu} \) path, 
DLG can still accurately capture the preference encoded in the gradient \( \nabla W_i^\mu \) through Eq.~(\ref{eq:closed_form_vae}), enabling label reconstruction through this deterministic path.

\begin{table*}[t]
  \centering
  \scriptsize
  \setlength{\tabcolsep}{3.1pt}
  \renewcommand{\arraystretch}{0.5}
  {
  \begin{tabular}{ccc cccc cccc cccc}
  \toprule
  \multirow{2}{*}{Dataset} & \multirow{2}{*}{Attributes} & \multirow{2}{*}{Method} 
  & \multicolumn{4}{c}{FedNCF} 
  & \multicolumn{4}{c}{FedVAE} 
  & \multicolumn{4}{c}{FedRAP} \\
  \cmidrule(r){4-7} \cmidrule(r){8-11} \cmidrule(r){12-15}
   & & & HR@10 $\uparrow$ & NDCG@10 $\uparrow$ & F1 $\downarrow$ & BAcc $\downarrow$
         & HR@10 $\uparrow$ & NDCG@10 $\uparrow$ & F1 $\downarrow$ & BAcc $\downarrow$
         & HR@10 $\uparrow$ & NDCG@10 $\uparrow$ & F1 $\downarrow$ & BAcc $\downarrow$ \\
  \midrule
  
  \multirow{6}{*}{ML-100K}
    & \multirow{3}{*}{Gender}
      & Original & 0.6392 & 0.3634 & 0.5982 & 0.6068 & 0.6703 & 0.3910 & 0.5909 & 0.6097 & 0.4414 & 0.2385 & 0.5789 & 0.5853 \\
    & & APM      & 0.5495 & 0.3083 & 0.5668 & 0.5675 & 0.5183 & 0.2938 & 0.5710 & 0.5785 & 0.4121 & 0.2249 & \textbf{0.5411} & \textbf{0.5485} \\
    & & \method{}     & \textbf{0.6154} & \textbf{0.3447} & \textbf{0.4992} & \textbf{0.5026} 
                 & \textbf{0.6667} & \textbf{0.3730} & \textbf{0.4113} & \textbf{0.5295} 
                 & \textbf{0.4359} & \textbf{0.2473} & 0.5490 & 0.5508 \\
    \cmidrule{2-15}
    
    & \multirow{3}{*}{Age}
      & Original & 0.6392 & 0.3634 & 0.4590 & 0.4739 & 0.6703 & 0.3910 & 0.3956 & 0.4179 & 0.4414 & 0.2385 & 0.4185 & 0.4250 \\
    & & APM      & 0.5495 & 0.3083 & 0.3915 & 0.3945 & 0.5183 & 0.2938 & 0.3362 & 0.3843 & 0.4121 & 0.2249 & 0.3793 & 0.3844 \\
    & & \method{}     & \textbf{0.5916} & \textbf{0.3247} & \textbf{0.3556} & \textbf{0.3623} 
                 & \textbf{0.6392} & \textbf{0.3626} & \textbf{0.2105} & \textbf{0.3431} 
                 & \textbf{0.4322} & \textbf{0.2452} & \textbf{0.3601} & \textbf{0.3706} \\
  \midrule
  
  \multirow{6}{*}{ML-1M}
    & \multirow{3}{*}{Gender}
      & Original & 0.6662 & 0.3897 & 0.7024 & 0.7032 & 0.6896 & 0.4169 & 0.7189 & 0.7204 & 0.4365 & 0.2381 & 0.7069 & 0.7089 \\
    & & APM      & 0.5530 & 0.3079 & 0.6162 & 0.6179 & 0.5366 & 0.3005 & 0.6654 & 0.6666 & 0.4114 & 0.2268 & 0.6575 & 0.6621 \\
    & & \method{}     & \textbf{0.6232} & \textbf{0.3519} & \textbf{0.5176} & \textbf{0.5136} 
                 & \textbf{0.6697} & \textbf{0.3994} & \textbf{0.4191} & \textbf{0.5536} 
                 & \textbf{0.4429} & \textbf{0.2437} & \textbf{0.6285} & \textbf{0.6283} \\

    \cmidrule{2-15}
  
    & \multirow{3}{*}{Age}
      & Original & 0.6662 & 0.3897 & 0.5896 & 0.5926 & 0.6896 & 0.4169 & 0.6186 & 0.6181 & 0.4365 & 0.2381 & 0.5488 & 0.5535 \\
    & & APM      & 0.5530 & 0.3079 & 0.4832 & 0.4850 & 0.5366 & 0.3005 & 0.5140 & 0.5169 & 0.4114 & 0.2268 & 0.5109 & 0.5159 \\
    & & \method{}     & \textbf{0.6360} & \textbf{0.3595} & \textbf{0.3641} & \textbf{0.3669} 
                 & \textbf{0.6779} & \textbf{0.4017} & \textbf{0.1755} & \textbf{0.3347} 
                 & \textbf{0.4447} & \textbf{0.2458} & \textbf{0.4012} & \textbf{0.4000} \\

  \midrule
  
  \multirow{6}{*}{LastFM}
    & \multirow{3}{*}{Gender}
      & Original & 0.7158 & 0.4411 & 0.6907 & 0.6926 & 0.7407 & 0.4648 & 0.6806 & 0.6895 & 0.4763 & 0.2697 & 0.6909 & 0.6912 \\
    & & APM      & 0.5744 & 0.3266 & 0.5950 & 0.6010 & 0.5818 & 0.3285 & 0.6463 & 0.6560 & 0.4351 & 0.2516 & \textbf{0.6079} & \textbf{0.6079} \\
    & & \method{}     & \textbf{0.6475} & \textbf{0.3867} & \textbf{0.4410} & \textbf{0.5047} 
                 & \textbf{0.7279} & \textbf{0.4566} & \textbf{0.3713} & \textbf{0.5222} 
                 & \textbf{0.4783} & \textbf{0.2689} & 0.6087 & 0.6088 \\

    \cmidrule{2-15}
  
    & \multirow{3}{*}{Age}
      & Original & 0.7158 & 0.4411 & 0.4775 & 0.4813 & 0.7407 & 0.4648 & 0.5329 & 0.5391 & 0.4763 & 0.2697 & 0.5098 & 0.5165 \\
    & & APM      & 0.5744 & 0.3266 & 0.4174 & 0.4189 & 0.5818 & 0.3285 & 0.4677 & 0.4682 & 0.4351 & 0.2516 & 0.4401 & 0.4405 \\
    & & \method{}     & \textbf{0.6584} & \textbf{0.3932} & \textbf{0.3649} & \textbf{0.3685} 
                 & \textbf{0.7358} & \textbf{0.4610} & \textbf{0.2033} & \textbf{0.3556} 
                 & \textbf{0.4794} & \textbf{0.2709} & \textbf{0.4301} & \textbf{0.4321} \\

  \bottomrule
  \end{tabular}
  }
\caption{Performance (HR@10, NDCG@10) and privacy leakage (F1, BAcc) under different methods across datasets and attributes. Bold indicates the best result (excluding Original). All results are averaged over three independent runs.
Due to the space limit, we report the results of HR@$k$ and NDCG@$k$ in Appendix C, where \( k \in \{5, 15, 20\} \).}
  \label{tab:fedrec_exp}
\end{table*}

\subsubsection{MLP with DSVAE}
Our proposed DSVAE enhances VAE by injecting extra stochasticity into the \( \boldsymbol{\mu} \) path, introducing perturbations to both \( \hat{y}_i' \) and \( \delta_i \), which jointly block DLG from accurately reconstructing user labels.
To achieve this, DSVAE injects stochasticity directly into \( \boldsymbol{\mu} \) by redefining the latent variable as:

\begin{equation}
    \mathbf{u} = \boldsymbol{\mu} \odot \boldsymbol{\epsilon}_1' + \boldsymbol{\sigma} \odot \boldsymbol{\epsilon}_2', \quad 
    \boldsymbol{\epsilon}_1' \sim \mathcal{N}(\mathbf{1}, \boldsymbol{\lambda}), \quad 
    \boldsymbol{\epsilon}_2' \sim \mathcal{N}(\mathbf{0}, \mathbf{I}),
\end{equation}
where \( \boldsymbol{\lambda} \) is a tunable hyperparameter that controls the strength of stochasticity.

\begin{corollary}
    \label{corollary:closed_form_dsvae}
The correction term \( \delta_i \) under stochastic perturbation derived from \( \nabla W_i^\mu \) is given by:
\begin{equation}
    \delta_i = \frac{(\mathbf{z}' \odot \boldsymbol{\epsilon}_1')^\top \nabla W_i^\mu}{\| \mathbf{z}' \odot \boldsymbol{\epsilon}_1' \|^2}.
    \label{eq:closed_form_dsv}
\end{equation}
\end{corollary}

\begin{proof}
    The proof can be found in Appendix B.4.
\end{proof}

As \( y_i \), \( \hat{y}_i \), \( \mathbf{z} \), and \( \boldsymbol{\epsilon}_1 \) are generated locally on the client during training, they are independent of server-side data.
Given that the gradient computes as \( \nabla W_i^\mu = (\hat{y}_i - y_i)(\mathbf{z} \odot \boldsymbol{\epsilon}_1) \), substituting into Eq.~(\ref{eq:closed_form_dsv}) yields:

\begin{equation}
    \delta_i = (\boldsymbol{\epsilon}_1^\top \boldsymbol{\epsilon}_1') \odot \frac{(\hat{y}_i - y_i)(\mathbf{z}^\top \mathbf{z}')}{\| \mathbf{z}' \odot \boldsymbol{\epsilon}_1' \|^2}.
    \label{eq:closed_form_var}
\end{equation}

For \( \hat{y}_i' \), building upon the VAE framework, DSVAE similarly disrupts the direct mapping to the client-side prediction \( \hat{y}_i \).
For \( \delta_i \), the injected noise into the \( \boldsymbol{\mu} \) path introduces a multiplicative term \( \boldsymbol{\epsilon}_1^\top \boldsymbol{\epsilon}_1' \), which effectively perturbs the information embedded in \( \nabla W_i^\mu \).
As a result, DSVAE introduces dual stochasticity.  
This design effectively prevents DLG from accurately reconstructing user labels.
In addition, for iDLG~\cite{zhao2020idlg}, the introduction of \( \boldsymbol{\epsilon}_1 \) invalidates its label inference mechanism based on the sign of gradients.

\section{Experiments}\label{sec:exp}
\subsection{Experimental Setup}
\subsubsection{Datasets}

We conduct experiments on three widely-used real-world datasets, each containing user-item interactions and user attributes (e.g., age and gender):

\begin{itemize}[leftmargin=*]\setlength{\itemsep}{-\itemsep}
    \item \textbf{ML-100K}: 100K movie ratings from 1,000 users on 1,700 movies, with user demographics including gender, age, and occupation.
    \item \textbf{ML-1M}: 1M ratings from 6,040 users on 3,706 movies, with similar user attribute information as ML-100K.
    \item \textbf{LastFM-360K}: User-artist interactions and user profiles (gender, age, country) from a music streaming platform.
\end{itemize}

\subsubsection{Recommendation Models}

We evaluate our proposed framework on two representative 
and one personalized FedRecs models:

\begin{itemize}[leftmargin=*]\setlength{\itemsep}{-\itemsep}
    \item \textbf{FedNCF}: A federated version of Neural Collaborative Filtering~\cite{perifanis2022federated}.
    \item \textbf{FedVAE}: A federated adaptation of MultVAE that incorporates an adaptive learning rate~\cite{polato2021federated}.
    \item \textbf{FedRAP}: A personalized FedRecs model combining global item embeddings with user-specific residuals~\cite{li2023federated}.
\end{itemize}

\subsubsection{Unlearning Methods}

As group-level methods~\cite{hu2024user, wu2025aegis} cannot apply to user-level FedRecs, we compare \method{} with the user-level baseline and the standard training strategy.
\begin{itemize}[leftmargin=*]\setlength{\itemsep}{-\itemsep}
    \item \textbf{Original}: The model without attribute unlearning.
    \item \textbf{APM}: A DP-based method that perturbs model parameters during training~\cite{zhang2023comprehensive}.
\end{itemize}

\subsubsection{Attack Setting}

Following
~\cite{zhang2023comprehensive}, we assume an honest-but-curious server that attempts to infer user attributes from gradients, along with an external adversary that exploits user embeddings for attribute inference. 

\subsubsection{Evaluation Metrics}

We use micro-averaged F1 score and Balanced Accuracy (BAcc) to evaluate attribute unlearning effectiveness.  
Gradient unlearning effectiveness 
is evaluated using accuracy.  
Recommendation performance is evaluated using Hit Ratio (HR) and Normalized Discounted Cumulative Gain (NDCG), under the leave-one-out evaluation protocol. We truncate the ranked list at 5, 10, 15, and 20.
More details of experimental setup, including data pre-processing, attack settings, training parameter settings, and hardware information, are provided in Appendix A.

\subsection{Results and Discussions}
\subsubsection{Attribute Unlearning Performance}

We report the F1 score and BAcc in Table~\ref{tab:fedrec_exp}.
On FedNCF and FedVAE, our proposed \method{} achieves an average BAcc reduction of 26.42\%, while APM achieves only 11.5\%, consistently outperforming baselines in both F1 and BAcc across all datasets and attributes.
For FedRAP, the average BAcc reduction is 14.09\% for \method{} and 9.24\% for APM.
FedRAP shows relatively weak unlearning performance, largely because its residual-based user embedding is hard to modify. 
To improve efficiency, we aggregate the residual into a single dimension during adversarial training, which may further limit its attribute unlearning ability.

\subsubsection{Recommendation Performance}
 As shown in Table~\ref{tab:fedrec_exp}, \method{} and APM reduce NDCG@10 by 4.51\% and 20.05\%, respectively, on average.  
Across all settings, \method{} consistently outperforms APM in both NDCG and HR, demonstrating better utility preservation while effectively unlearning sensitive attributes.
Interestingly, for FedRAP, \method{} even leads to a slight improvement in recommendation performance, which may be attributed to the unique structure of its user embedding.

\subsubsection{SUT Analysis}
To assess the effect of SUT on training stability, we show the recommendation and unlearning performance under different adversarial strategies in Fig.~\ref{fig:sut}. 
\textit{Global} applies adversarial training throughout all epochs, while \textit{Pretrain} pre-trains the adversarial model for multiple epochs (10 for FedRAP, 50 for FedVAE, and 100 for FedNCF) before unlearning.
As shown in Fig.~\ref{fig:sut}, SUT yields the best recommendation performance, while \textit{Global} performs the worst, indicating that SUT's adaptive perturbation prevents excessive disruption to user embeddings.
For unlearning performance, SUT and \textit{Global} perform similarly, but \textit{Pretrain} performs worst due to limited adversarial model generalization under stochastic sampling.
Overall, SUT better balances recommendation and unlearning performance.

\subsubsection{DSVAE Analysis}
To validate the effectiveness of DSVAE, we present the gradient unlearning performance,
the component analysis of Eq.~(\ref{eq:closed_form_yi_dummy}), 
and the impact of the stochasticity coefficient \( \lambda \) on model performance. 
\begin{itemize}[leftmargin=*]\setlength{\itemsep}{-\itemsep}
    \item \textbf{Gradient Unlearning.} As shown in Fig.~\ref{fig:dsvae_1}, the original MLP is vulnerable to gradient attacks, which accurately infer user attributes, achieving up to 90\% accuracy in FedVAE.
While VAE offers slight resistance to the attack, the adversary can still achieve up to 89\% accuracy. 
In contrast, DSVAE significantly mitigates the attack across all three models through dual stochasticity, bringing the accuracy down to 58\% in FedVAE.
    \item \textbf{Component Analysis.} As shown in Fig.~\ref{fig:dsvae_2}, the original \( \hat{y}_i' \) remains slightly above random guessing, reflecting that it approximates the client-side prediction.
Both VAE and DSVAE are able to reduce \( \hat{y}_i' \) to the level of random guessing.  
However, VAE has limited effect on \( \delta_i \), resulting in a gradient's preference similar to the original MLP and leading to insufficient perturbation of \( y_i^* \).
In contrast, DSVAE significantly suppresses \( \delta_i \), masking the gradient's preference, resulting in effective perturbation of \( y_i^* \) and robust defense against gradient-based reconstruction.
    \item \textbf{Stochasticity Coefficient.} As shown in Fig.~\ref{fig:dsvae_3}, increasing \( \lambda \) introduces more stochasticity into the adversarial model, which weakens its generalization ability to attribute unlearning. 
This results in worse unlearning but better recommendation performance.
Meanwhile, the increased stochasticity also improves gradient unlearning performance, gradually approaching random guessing.
\end{itemize}

\subsubsection{Overhead Analysis}
SUT performs adversarial training only when the adversarial prediction is correct, reducing gradient computation. It is client-side and parameter-free, adding no memory or communication overhead. 
DSVAE replaces the last layer with two projections, doubling its parameters and causing a linear increase in computational, memory, and communication overhead. 
The combination of SUT and DSVAE slightly reduces the overall computational cost. We empirically evaluate our method and report the results in Appendix D.

\begin{figure}[t]
  \centering
  \includegraphics[width=0.8\linewidth]{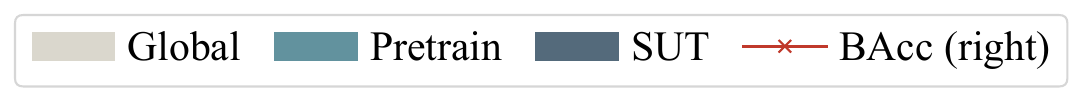} \\
  \includegraphics[width=\linewidth]{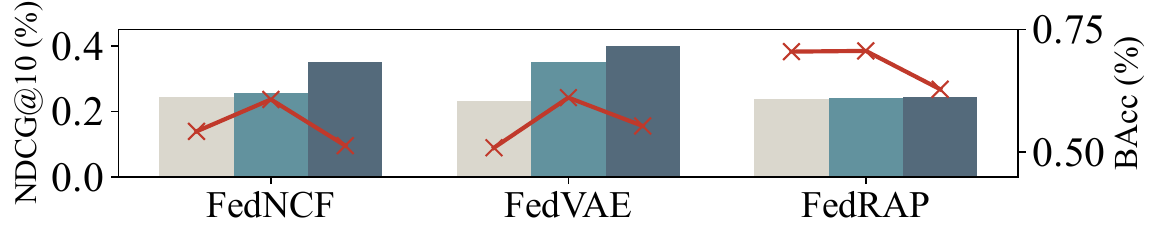}
  \caption{Ablation analysis of SUT. Recommendation (NDCG@10 $\uparrow$) and unlearning (BAcc $\downarrow$) performance under different adversarial training strategies on ML-1M (gender).}
  \label{fig:sut}
\end{figure}

\begin{figure}[t]
  \centering
  \hspace{0.9cm}\includegraphics[width=0.8\linewidth]{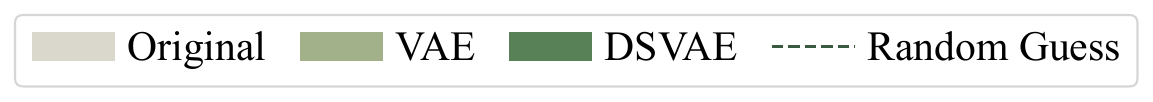} \\
  \vspace{-6pt}
  \subfigure[Gradient unlearning performance]{\label{fig:dsvae_1}
    \vspace{-10pt}\includegraphics[width=\linewidth]{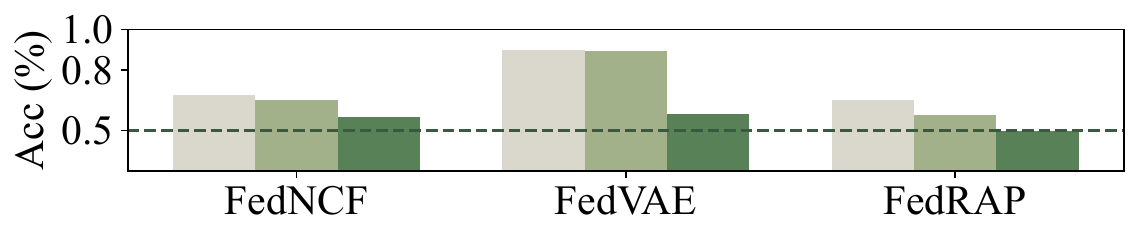}
  }
  \subfigure[Component Analysis of Reconstruction]{\label{fig:dsvae_2}
    \includegraphics[width=\linewidth]{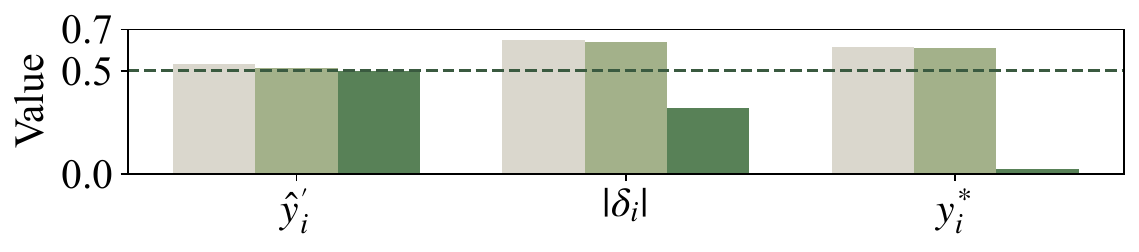}
  }
  \subfigure[Impact of the stochasticity coefficient \( \lambda \)]{\label{fig:dsvae_3}
  \includegraphics[width=\linewidth]{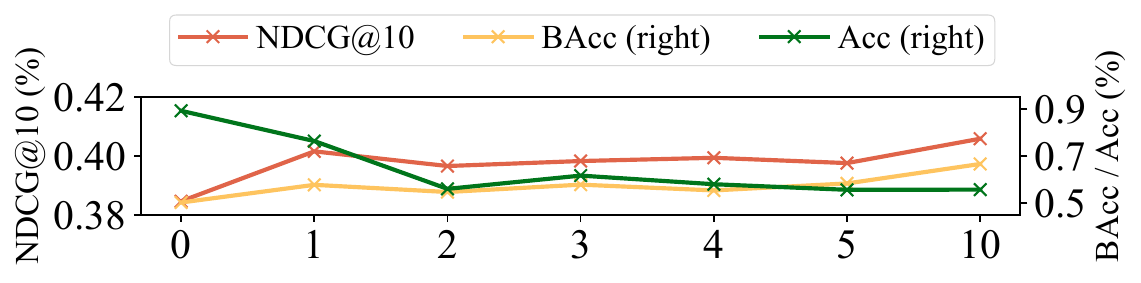}
  }
  \caption{Ablation analysis of DSVAE, 
  conducted on ML-1M (gender). 
  (a) Gradient unlearning performance across three models. 
  (b) Component reconstruction analysis in FedVAE. 
  (c) Effect of the stochasticity coefficient \( \lambda \) on recommendation (NDCG@10 \( \uparrow \)), attribute unlearning (BAcc \( \downarrow \)), and gradient unlearning (Acc \( \downarrow \)) performance in FedVAE.}  
    \label{fig:dsvae}
\end{figure}

\section{Conclusion}
In this paper, we study attribute unlearning in user-level FedRecs, aiming to remove sensitive attribute information from user embeddings while maintaining recommendation utility. 
User-level FedRecs are more challenging than group-level settings due to decentralized optimization and highly personalized data distributions.
We identify two key challenges for adversarial training-based unlearning in this setting: i) training instability caused by user heterogeneity, and ii) gradient-based leakage of sensitive information. 
To tackle these challenges, we propose \method{}, which incorporates two core components: an adaptive adversarial training strategy (SUT) that dynamically adjusts the perturbation budget based on local optimization signals, and a dual-stochastic variational autoencoder (DSVAE) that effectively masks attribute information in the gradient. 
Extensive experiments on three real-world datasets and multiple representative FedRecs models demonstrate that \method{} achieves significantly better unlearning effectiveness and recommendation performance compared to existing baselines. 
Our findings show that adaptive adjustment is crucial for stable adversarial training in user-level FedRecs, and that label-side defenses are essential against gradient-based attacks.

\appendix

\section{Experimental Details}\label{appendix_experimental_details}
\subsection{Dataset pre-processing}
For ML-100K and ML-1M, we retain only users who have interacted with at least 5 items, and items that have received interactions from at least 5 users. 
 For LastFM, we retain users with at least 23 interactions and items with at least 51 user interactions.
After filtering, we balance the number of users across gender groups (male and female).
To evaluate recommendation performance, we reserve each user's most recent interaction (based on timestamp) for testing.
The age attribute is discretized into three categories: under 27, 27 to 38, and over 38 for ML-100K; under 25, 25 to 31, and over 31 for ML-1M; and under 21, 21 to 26, and over 26 for KuaiSAR. 
The gender attribute is restricted to male and female categories.

\subsection{Attack settings}
Following prior work~\cite{zhang2023comprehensive}, we assume an honest-but-curious server, along with an external adversary equipped with the following capabilities:
(i) the adversary has access to user embeddings $\mathbf{e}_m$ at any iteration $t$ during training;  %
(ii) the adversary knows the sensitive attributes of a small subset of users (e.g., 20\%), who may collaborate with the server or have shared their information voluntarily.
(iii) the server stores the latest uploaded gradients from users and launches a gradient-based attack at a specific iteration.

For attribute inference attacks, we implement the adversary as a MLP consisting of one hidden layer with 100 units and a softmax output layer. We use 20\% of the users as the training set and evaluate the attack performance on the remaining 80\%.
Specifically for FedRAP, we design an ensemble attack that targets the user-specific residuals. Given that the residual matrix $\mathbf{e}_m \in \mathbb{R}^{M \times N}$ (where $M$ is the number of items and $N$ is the embedding dimension), 
we treat each feature dimension independently and train a separate MLP classifier for each column. The final prediction is obtained by averaging the outputs from all $N$ classifiers.

For gradient-based attacks, we adopt the standard DLG method~\cite{zhu2019deep}, where the server is assumed to have access to both the model parameters and user gradients. 
To enable this, the server stores each client's latest gradient and the corresponding model state. 
After 100 global epochs, the DLG attack is launched using the recorded information.

\subsection{Training parameters}

For all FedRecs, we set the number of local epochs to 1 and sample 10\% of users for training in each global round. 
For FedNCF, we set the latent dimension, local learning rate, and number of global epochs to 32, 0.1, and 800, respectively.  
For FedVAE, they are set to 200, 0.01, and 400.  
For FedRAP, they are set to 32, 0.01, and 200.
All model parameters are initialized using Xavier initialization~\cite{glorot2010understanding}. 
In our proposed \method{}, we use a two-layer MLP as the adversarial model, with a hidden layer dimension of 100 and a softmax output layer.
The learning rate of the adversarial model is set to 0.1 for FedNCF and FedRAP, 0.001 for FedVAE.
The gradient scale in the GRL layer is set to 400, and the stochasticity coefficient \( \lambda \) in the DSVAE module is set to 4.

\subsection{Hardware information}
All models and algorithms are implemented using Python 3.10.0 and Pytorch 2.6.0. 
The experiments are conducted on a server running Ubuntu 18.04, equipped with 125GB of RAM and an NVIDIA GeForce RTX 4090 GPU.

\section{Proof of Theorems}
\subsection{Proof of Theorem \ref{theorem:closed_yi}}
\label{sec:proof_yi}
We provide the proof of Theorem~\ref{theorem:closed_yi}, beginning by introducing several intermediate variables.
\begin{proof}
Given a class \( i \),
let \( W_i' \) denote the dummy weight vector from the final layer;  
\( s_i' = f(W_i') \), the dummy logit derived from \( W_i' \);  
\( \hat{y}_i' = \mathrm{Softmax}(s_i') \), the predicted dummy probability.    
The dummy cross-entropy loss is defined as:
\[
\mathcal{L}_{\text{CE}}' = - \sum_{j=1}^{n} y_j' \log \hat{y}_j',
\]
We reconstruct \( y_i' \) by minimizing the squared gradient mismatch.
\begin{equation*}
    \begin{aligned}
        \min_{y_i'} \left\| \nabla W_i - \nabla W_i' \right\|_2^2
        &= \min_{y_i'} \left\| \nabla W_i 
        - \frac{\partial \mathcal{L}_{\text{CE}}'}{\partial W_i'} \right\|_2^2 \\
        &= \min_{y_i'} \left\| \nabla W_i 
        - \frac{\partial \mathcal{L}_{\text{CE}}'}{\partial s_i'} \cdot \frac{\partial s_i'}{\partial W_i'}  \right\|_2^2 \\
        &= \min_{y_i'} \left\| \nabla W_i 
        - (\hat{y}_i' - y_i') \cdot \frac{\partial s_i'}{\partial W_i'} \right\|_2^2.
    \end{aligned}
\end{equation*}
Solving the first-order condition of this convex quadratic yields:
\[
y_i^* = \hat{y}_i' - \frac{
\left( \dfrac{\partial s_i'}{\partial W_i'} \right)^\top \nabla W_i
}{
\left\| \dfrac{\partial s_i'}{\partial W_i'} \right\|^2
}.
\]
We define the correction term \( \delta_i \) as a function of \( \nabla W_i \), yielding the closed-form reconstruction:
\[
y_i^* = \hat{y}_i' - \delta_i, \quad
\delta_i := \mathcal{G}(\nabla W_i) := \frac{
\left( \dfrac{\partial s_i'}{\partial W_i'} \right)^\top \nabla W_i
}{
\left\| \dfrac{\partial s_i'}{\partial W_i'} \right\|^2
}.
\]
\end{proof}

\begin{table*}[t]
    \centering
    \caption{Additional recommendation performance (HR@$k$ and NDCG@$k$ for \( k \in \{5, 15, 20\} \)) under different methods across datasets and attributes.}
    \label{tab:addition_exp}
    \resizebox{\linewidth}{!}{
    \begin{tabular}{ccc cccccc | cccccc | cccccc}
    \toprule
    \multirow{2}{*}{Dataset} & \multirow{2}{*}{Attributes} & \multirow{2}{*}{Method} 
    & \multicolumn{6}{c}{FedNCF} 
    & \multicolumn{6}{c}{FedVAE} 
    & \multicolumn{6}{c}{FedRAP} \\
    \cmidrule(r){4-9} \cmidrule(r){10-15} \cmidrule(r){16-21}
     & & & HR@5 & NDCG@5 & HR@15 & NDCG@15 & HR@20 & NDCG@20
     & HR@5 & NDCG@5 & HR@15 & NDCG@15 & HR@20 & NDCG@20
     & HR@5 & NDCG@5 & HR@15 & NDCG@15 & HR@20 & NDCG@20 \\
    \midrule
    
    \multirow{6}{*}{ML-100K}
      & \multirow{3}{*}{Gender}
        & Original & 0.4396 & 0.2992 & 0.7491 & 0.3924 & 0.8260 & 0.4106
                  & 0.4615 & 0.3240 & 0.7875 & 0.4219 & 0.8352 & 0.4331
                  & 0.2912 & 0.1907 & 0.5183 & 0.2590 & 0.5604 & 0.2690 \\
      & & APM      & 0.3773 & 0.2521 & 0.6722 & 0.3406 & 0.7546 & 0.3601
                  & 0.3718 & 0.2464 & 0.6429 & 0.3268 & 0.6941 & 0.3390
                  & 0.2784 & 0.1825 & 0.4725 & 0.2407 & 0.5147 & 0.2508 \\
      & & \method{}     & 0.4322 & 0.2852 & 0.7179 & 0.3720 & 0.7930 & 0.3898
                  & 0.4432 & 0.3022 & 0.7527 & 0.3959 & 0.8242 & 0.4128
                  & 0.3040 & 0.2126 & 0.5110 & 0.2745 & 0.5842 & 0.2918 \\
      \cmidrule{2-21}
      
      & \multirow{3}{*}{Age}
      & Original & 0.4396 & 0.2992 & 0.7491 & 0.3924 & 0.8260 & 0.4106
                 & 0.4615 & 0.3240 & 0.7875 & 0.4219 & 0.8352 & 0.4331
                 & 0.2912 & 0.1907 & 0.5183 & 0.2590 & 0.5604 & 0.2690 \\
      & & APM      & 0.3773 & 0.2521 & 0.6722 & 0.3406 & 0.7546 & 0.3601
                 & 0.3718 & 0.2464 & 0.6429 & 0.3268 & 0.6941 & 0.3390
                 & 0.2784 & 0.1825 & 0.4725 & 0.2407 & 0.5147 & 0.2508 \\
      & & \method{}     & 0.3846 & 0.2584 & 0.6923 & 0.3511 & 0.7674 & 0.3689
                 & 0.4322 & 0.2968 & 0.7326 & 0.3872 & 0.8040 & 0.4041
                 & 0.2930 & 0.2001 & 0.5293 & 0.2711 & 0.5934 & 0.2863 \\
    \midrule
    
    \multirow{6}{*}{ML-1M}
      & \multirow{3}{*}{Gender}
      & Original & 0.4865 & 0.3317 & 0.7657 & 0.4160 & 0.8233 & 0.4296
                 & 0.5234 & 0.3630 & 0.7803 & 0.4410 & 0.8283 & 0.4523
                 & 0.2993 & 0.1940 & 0.5184 & 0.2598 & 0.5784 & 0.2740 \\
      & & APM      & 0.3812 & 0.2522 & 0.6653 & 0.3377 & 0.7414 & 0.3556
                 & 0.3754 & 0.2482 & 0.6448 & 0.3291 & 0.7241 & 0.3478
                 & 0.2891 & 0.1872 & 0.4845 & 0.2462 & 0.5339 & 0.2578 \\
      & & \method{}     & 0.4441 & 0.2944 & 0.7209 & 0.3778 & 0.7917 & 0.3945
                 & 0.5044 & 0.3438 & 0.7575 & 0.4210 & 0.8207 & 0.4359
                 & 0.3087 & 0.2007 & 0.5252 & 0.2655 & 0.5711 & 0.2763 \\
  
      \cmidrule{2-21}
    
      & \multirow{3}{*}{Age}
      & Original & 0.4865 & 0.3317 & 0.7657 & 0.4160 & 0.8233 & 0.4296
                 & 0.5234 & 0.3630 & 0.7803 & 0.4410 & 0.8283 & 0.4523
                 & 0.2993 & 0.1940 & 0.5184 & 0.2598 & 0.5784 & 0.2740 \\
      & & APM      & 0.3812 & 0.2522 & 0.6653 & 0.3377 & 0.7414 & 0.3556
                 & 0.3754 & 0.2482 & 0.6448 & 0.3291 & 0.7241 & 0.3478
                 & 0.2891 & 0.1872 & 0.4845 & 0.2462 & 0.5339 & 0.2578 \\
      & & \method{}     & 0.4500 & 0.2995 & 0.7323 & 0.3849 & 0.7978 & 0.4003
                 & 0.5038 & 0.3452 & 0.7633 & 0.4243 & 0.8248 & 0.4388
                 & 0.3116 & 0.2031 & 0.5366 & 0.2701 & 0.6173 & 0.2891 \\
  
    \midrule
    
    \multirow{6}{*}{LastFM}
      & \multirow{3}{*}{Gender}
      & Original & 0.5450 & 0.3860 & 0.7925 & 0.4615 & 0.8496 & 0.4750
                 & 0.5773 & 0.4118 & 0.8206 & 0.4860 & 0.8672 & 0.4970
                 & 0.3360 & 0.2245 & 0.5533 & 0.2901 & 0.6012 & 0.3015 \\
      & & APM      & 0.3999 & 0.2701 & 0.6678 & 0.3512 & 0.7431 & 0.3691
                 & 0.4078 & 0.2722 & 0.6874 & 0.3565 & 0.7564 & 0.3729
                 & 0.3250 & 0.2158 & 0.4941 & 0.2672 & 0.5397 & 0.2780 \\
      & & \method{}     & 0.4778 & 0.3319 & 0.7327 & 0.4093 & 0.7927 & 0.4235
                 & 0.5647 & 0.4036 & 0.8095 & 0.4782 & 0.8582 & 0.4897
                 & 0.3378 & 0.2235 & 0.5553 & 0.2894 & 0.6073 & 0.3017 \\
  
      \cmidrule{2-21}
    
      & \multirow{3}{*}{Age}
      & Original & 0.5450 & 0.3860 & 0.7925 & 0.4615 & 0.8496 & 0.4750
                 & 0.5773 & 0.4118 & 0.8206 & 0.4860 & 0.8672 & 0.4970
                 & 0.3360 & 0.2245 & 0.5533 & 0.2901 & 0.6012 & 0.3015 \\
      & & APM      & 0.3999 & 0.2701 & 0.6678 & 0.3512 & 0.7431 & 0.3691
                 & 0.4078 & 0.2722 & 0.6874 & 0.3565 & 0.7564 & 0.3729
                 & 0.3250 & 0.2158 & 0.4941 & 0.2672 & 0.5397 & 0.2780 \\
      & & \method{}     & 0.4884 & 0.3384 & 0.7474 & 0.4168 & 0.8059 & 0.4306
                 & 0.5700 & 0.4072 & 0.8159 & 0.4822 & 0.8640 & 0.4935
                 & 0.3419 & 0.2265 & 0.5613 & 0.2926 & 0.6131 & 0.3048 \\
  
    \bottomrule
    \end{tabular}
    }
\end{table*}

\begin{table}[t]
    \centering
    \caption{Computation overhead comparison  across different components. 
    ``Fwd'' and ``Bwd'' respectively denote the forward and backward propagation time.}
    \label{tab:time_overhead}
    \resizebox{0.65\linewidth}{!}{
    \begin{tabular}{lcccc}
        \toprule
        Time / ms & Adv & +SUT & +DSVAE & \method{} \\
        \midrule
        Fwd & 2.953 & 2.876 & 3.033 & 2.964 \\
        Bwd & 5.536 & 5.112 & 5.640 & 5.117 \\
        Overall & 8.489 & 7.988 & 8.673 & 8.081 \\
        \bottomrule
    \end{tabular}
    }
\end{table}

\subsection{Proof of Corollary \ref{corollary:closed_form_ori}}
We provide a proof of Corollary~\ref{corollary:closed_form_ori} as a case of Theorem~\ref{theorem:closed_yi} applied to the original MLP.
\begin{proof}
\label{sec:proof_ori}
In a standard MLP, \( s_i' = W_i'^\top \mathbf{h}' \), where \( \mathbf{h}' \) is the dummy input to the final layer. Since \( \frac{\partial s_i'}{\partial W_i'} = \mathbf{h}' \), we obtain the correction term:
\[
\delta_i = \frac{(\mathbf{h}')^\top \nabla W_i}{\|\mathbf{h}'\|^2}.
\]
\end{proof}

\subsection{Proof of Corollary \ref{corollary:closed_form_vae}}
\label{sec:proof_vae}
We provide a proof of Corollary~\ref{corollary:closed_form_vae} as a case of Theorem~\ref{theorem:closed_yi} applied to the MLP with VAE.
\begin{proof}
In the MLP with VAE, the logit \( s_i' \) is computed as:
\[
s_i' = \left( \left(W_i^{\mu}\right)'^\top + \left(W_i^{\sigma}\right)'^\top \odot \boldsymbol{\epsilon}' \right) \mathbf{z}',
\]
where \( \mathbf{z}' \) is the VAE input, and \( \boldsymbol{\epsilon}' \) is Gaussian sampled.
\( \left(W_i^{\mu}\right)' \) and \( \left(W_i^{\sigma}\right)' \) are weights of the \( \boldsymbol{\mu} \) and \( \boldsymbol{\sigma} \) paths, respectively.
The dummy label \( y_i' \) is derived from \( \nabla W_i^{\mu} \) along the \( \boldsymbol{\mu} \) path as:
\[
\begin{aligned}
\min_{y_i'} \left\| \nabla W_i^{\mu} - \nabla \left(W_i^{\mu}\right)' \right\|_2^2 
&= \min_{y_i'} \left\| \nabla W_i^{\mu} - (\hat{y}_i' - y_i') \cdot \frac{\partial s_i'}{\partial \left(W_i^{\mu}\right)'} \right\|_2^2 \\
&= \min_{y_i'} \left\| \nabla W_i^{\mu} - (\hat{y}_i' - y_i') \cdot \mathbf{z}' \right\|_2^2.
\end{aligned}
\]
This yields the closed-form:
\[
y_i^* = \hat{y}_i' - \delta_i, \quad \text{where} \quad \delta_i = \frac{(\mathbf{z}')^\top \nabla W_i^{\mu}}{\|\mathbf{z}'\|^2}.
\]
\end{proof}

\subsection{Proof of Corollary \ref{corollary:closed_form_dsvae}}
\label{sec:proof_dsvae}
We provide a proof of Corollary~\ref{corollary:closed_form_dsvae} as a case of Theorem~\ref{theorem:closed_yi} applied to the MLP with DSVAE.
\begin{proof}
In the MLP with DSVAE, the logit \( s_i' \) is computed as:
\[
s_i' = \left( \left(W_i^{\mu}\right)'^\top \odot \boldsymbol{\epsilon}_1' + \left(W_i^{\sigma}\right)'^\top \odot \boldsymbol{\epsilon}_2' \right) \mathbf{z}',
\]
where \( \boldsymbol{\epsilon}_1' \) and \( \boldsymbol{\epsilon}_2' \) are independent Gaussian samples for the \( \boldsymbol{\mu} \) and \( \boldsymbol{\sigma} \) paths.
Dummy label \( y_i' \) is inferred along the \( \boldsymbol{\mu} \) path based on \( \nabla W_i^{\mu} \), following:
\[
\begin{aligned}
\min_{y_i'} \left\| \nabla W_i^{\mu} - \nabla \left(W_i^{\mu}\right)' \right\|_2^2 
&= \min_{y_i'} \left\| \nabla W_i^{\mu} - (\hat{y}_i' - y_i') \cdot \frac{\partial s_i'}{\partial \left(W_i^{\mu}\right)'} \right\|_2^2 \\
&= \min_{y_i'} \left\| \nabla W_i^{\mu} - (\hat{y}_i' - y_i') \cdot \left( \boldsymbol{\epsilon}_1' \odot \mathbf{z}' \right) \right\|_2^2. 
\end{aligned}
\]
This yields the closed-form:
\[
y_i^* = \hat{y}_i' - \delta_i, \quad \text{where} \quad \delta_i = \frac{ \left( \boldsymbol{\epsilon}_1' \odot \mathbf{z}' \right)^\top \nabla W_i^{\mu} }{ \left\| \boldsymbol{\epsilon}_1' \odot \mathbf{z}' \right\|^2 }.
\]
\end{proof}

\section{Additional Experimental Results}
\label{appendix_experimental_results}
We provide additional experimental results on recommendation performance (HR@$k$ and NDCG@$k$, for \( k \in \{5, 15, 20\} \)) under different methods across datasets and attributes, as shown in Table~\ref{tab:addition_exp}.

\section{Overhead Analysis}\label{appendix_overhead}
Table~\ref{tab:time_overhead} compares the computational overhead introduced by different components, evaluated on ML-1M with FedVAE. SUT significantly reduces the backward propagation time, while DSVAE slightly increases computation overhead. Putting them together, \method{} still achieves a 4.8\% reduction compared with the original Adv.

\section*{Acknowledgments}
This work was supported in part by the National Natural Science Foundation of China under Grants (No.~62402148 and No.~62402418), the Zhejiang Province's 2025 ``Leading Goose + X'' Science and Technology Plan under Grant No.~2025C02034, the Key R\&D Program of Ningbo under No.~2024Z115, the Fundamental Research Funds for the Central Universities, and Ant Group Research Fund.

\bibliography{aaai2026}

@inproceedings{pbat2023su,
    author = {Su, Jiajie and Chen, Chaochao and Lin, Zibin and Li, Xi and Liu, Weiming and Zheng, Xiaolin},
    title = {Personalized Behavior-Aware Transformer for Multi-Behavior Sequential Recommendation},
    year = {2023},
    isbn = {9798400701085},
    publisher = {Association for Computing Machinery},
    address = {New York, NY, USA},
    url = {https://doi.org/10.1145/3581783.3611723},
    doi = {10.1145/3581783.3611723},
    booktitle = {Proceedings of the 31st ACM International Conference on Multimedia},
    pages = {6321–6331},
    numpages = {11},
    location = {Ottawa ON, Canada},
    series = {MM '23}
}

@article{duada25su,
    author = {Su, Jiajie and Chen, Chaochao and Wang, Yihao and Liu, Weiming and Li, Yuyuan and Wang, Tao and Li, Zhigang and Zheng, Xiaolin and Yin, Jianwei},
    title = {DuAda: Adaptive Targeted Model Poisoning Attack Framework via Dummy User Simulation on Federated Recommendation},
    year = {2025},
    issue_date = {November 2025},
    publisher = {Association for Computing Machinery},
    address = {New York, NY, USA},
    volume = {43},
    number = {6},
    issn = {1046-8188},
    url = {https://doi.org/10.1145/3757059},
    doi = {10.1145/3757059},
    journal = {ACM Trans. Inf. Syst.},
    month = sep,
    articleno = {161},
    numpages = {37}
}

@inproceedings{mcmahan2017communication,
  title={Communication-efficient learning of deep networks from decentralized data},
  author={McMahan, Brendan and Moore, Eider and Ramage, Daniel and Hampson, Seth and y Arcas, Blaise Aguera},
  booktitle={Artificial intelligence and statistics},
  pages={1273--1282},
  year={2017},
  organization={PMLR}
}

@article{chai2020secure,
  title={Secure federated matrix factorization},
  author={Chai, Di and Wang, Leye and Chen, Kai and Yang, Qiang},
  journal={IEEE Intelligent Systems},
  volume={36},
  number={5},
  pages={11--20},
  year={2020},
  publisher={IEEE}
}

@article{sun2022survey,
  title={A survey on federated recommendation systems},
  author={Sun, Zehua and Xu, Yonghui and Liu, Yong and He, Wei and Kong, Lanju and Wu, Fangzhao and Jiang, Yali and Cui, Lizhen},
  journal={arXiv preprint arXiv:2301.00767},
  year={2022}
}

@article{protection2018general,
  title={General data protection regulation (GDPR)},
  author={Protection, Formerly Data},
  journal={Intersoft Consulting, Accessed in October},
  volume={24},
  number={1},
  year={2018}
}

@article{illman2019california,
  title={California consumer privacy act},
  author={Illman, Erin and Temple, Paul},
  journal={The Business Lawyer},
  volume={75},
  number={1},
  pages={1637--1646},
  year={2019},
  publisher={JSTOR}
}

@inproceedings{ganhor2022unlearning,
  title={Unlearning protected user attributes in recommendations with adversarial training},
  author={Ganh{\"o}r, Christian and Penz, David and Rekabsaz, Navid and Lesota, Oleg and Schedl, Markus},
  booktitle={Proceedings of the 45th International ACM SIGIR Conference on Research and Development in Information Retrieval},
  pages={2142--2147},
  year={2022}
}

@inproceedings{li2023making,
  title={Making users indistinguishable: Attribute-wise unlearning in recommender systems},
  author={Li, Yuyuan and Chen, Chaochao and Zheng, Xiaolin and Zhang, Yizhao and Han, Zhongxuan and Meng, Dan and Wang, Jun},
  booktitle={Proceedings of the 31st ACM International Conference on Multimedia},
  pages={984--994},
  year={2023}
}

@inproceedings{hu2024user,
  title={User consented federated recommender system against personalized attribute inference attack},
  author={Hu, Qi and Song, Yangqiu},
  booktitle={Proceedings of the 17th ACM International Conference on Web Search and Data Mining},
  pages={276--285},
  year={2024}
}

@inproceedings{wu2025aegis,
  title={Aegis: Post-Training Attribute Unlearning in Federated Recommender Systems against Attribute Inference Attacks},
  author={Wu, Wenhan and Jiang, Jiawei and Hu, Chuang},
  booktitle={Proceedings of the ACM on Web Conference 2025},
  pages={3783--3793},
  year={2025}
}

@inproceedings{ganin2015unsupervised,
  title={Unsupervised domain adaptation by backpropagation},
  author={Ganin, Yaroslav and Lempitsky, Victor},
  booktitle={International conference on machine learning},
  pages={1180--1189},
  year={2015},
  organization={PMLR}
}

@article{feng2025raid,
  title={RAID: An In-Training Defense against Attribute Inference Attacks in Recommender Systems},
  author={Feng, Xiaohua and Li, Yuyuan and Yu, Fengyuan and Xiong, Ke and Fang, Junjie and Zhang, Li and Du, Tianyu and Chen, Chaochao},
  journal={arXiv preprint arXiv:2504.11510},
  year={2025}
}

@article{zhu2019deep,
  title={Deep leakage from gradients},
  author={Zhu, Ligeng and Liu, Zhijian and Han, Song},
  journal={Advances in neural information processing systems},
  volume={32},
  year={2019}
}

@article{zhao2020idlg,
  title={idlg: Improved deep leakage from gradients},
  author={Zhao, Bo and Mopuri, Konda Reddy and Bilen, Hakan},
  journal={arXiv preprint arXiv:2001.02610},
  year={2020}
}

@article{perifanis2022federated,
  title={Federated neural collaborative filtering},
  author={Perifanis, Vasileios and Efraimidis, Pavlos S},
  journal={Knowledge-Based Systems},
  volume={242},
  pages={108441},
  year={2022},
  publisher={Elsevier}
}

@inproceedings{polato2021federated,
  title={Federated variational autoencoder for collaborative filtering},
  author={Polato, Mirko},
  booktitle={2021 International Joint Conference on Neural Networks (IJCNN)},
  pages={1--8},
  year={2021},
  organization={IEEE}
}

@article{li2023federated,
  title={Federated recommendation with additive personalization},
  author={Li, Zhiwei and Long, Guodong and Zhou, Tianyi},
  journal={arXiv preprint arXiv:2301.09109},
  year={2023}
}

@article{zhang2023comprehensive,
  title={Comprehensive privacy analysis on federated recommender system against attribute inference attacks},
  author={Zhang, Shijie and Yuan, Wei and Yin, Hongzhi},
  journal={IEEE Transactions on Knowledge and Data Engineering},
  volume={36},
  number={3},
  pages={987--999},
  year={2023},
  publisher={IEEE}
}

@inproceedings{huang2020instahide,
  title={Instahide: Instance-hiding schemes for private distributed learning},
  author={Huang, Yangsibo and Song, Zhao and Li, Kai and Arora, Sanjeev},
  booktitle={International conference on machine learning},
  pages={4507--4518},
  year={2020},
  organization={PMLR}
}

@inproceedings{sun2021soteria,
  title={Soteria: Provable defense against privacy leakage in federated learning from representation perspective},
  author={Sun, Jingwei and Li, Ang and Wang, Binghui and Yang, Huanrui and Li, Hai and Chen, Yiran},
  booktitle={Proceedings of the IEEE/CVF conference on computer vision and pattern recognition},
  pages={9311--9319},
  year={2021}
}

@inproceedings{scheliga2022precode,
  title={Precode-a generic model extension to prevent deep gradient leakage},
  author={Scheliga, Daniel and M{\"a}der, Patrick and Seeland, Marco},
  booktitle={Proceedings of the IEEE/CVF Winter Conference on Applications of Computer Vision},
  pages={1849--1858},
  year={2022}
}

@inproceedings{ghazi2024labeldp,
  title={LabelDP-Pro: Learning with Label Differential Privacy via Projections},
  author={Ghazi, Badih and Huang, Yangsibo and Kamath, Pritish and Kumar, Ravi and Manurangsi, Pasin and Zhang, Chiyuan},
  booktitle={The Twelfth International Conference on Learning Representations},
  year={2024}
}

@article{struppek2023careful,
  title={Be careful what you smooth for: Label smoothing can be a privacy shield but also a catalyst for model inversion attacks},
  author={Struppek, Lukas and Hintersdorf, Dominik and Kersting, Kristian},
  journal={arXiv preprint arXiv:2310.06549},
  year={2023}
}

@article{he2024labobf,
  title={LabObf: A Label Protection Scheme for Vertical Federated Learning Through Label Obfuscation},
  author={He, Ying and Niu, Mingyang and Hua, Jingyu and Mao, Yunlong and Huang, Xu and Li, Chen and Zhong, Sheng},
  journal={arXiv preprint arXiv:2405.17042},
  year={2024}
}

@article{hasan2024based,
  title={Based recommender systems: a survey of approaches, challenges and future perspectives},
  author={Hasan, Emrul and Rahman, Mizanur and Ding, Chen and Huang, Jimmy and Raza, Shaina},
  journal={ACM Computing Surveys},
  year={2024},
  publisher={ACM New York, NY}
}

@inproceedings{qu2024towards,
  title={Towards personalized privacy: User-governed data contribution for federated recommendation},
  author={Qu, Liang and Yuan, Wei and Zheng, Ruiqi and Cui, Lizhen and Shi, Yuhui and Yin, Hongzhi},
  booktitle={Proceedings of the ACM Web Conference 2024},
  pages={3910--3918},
  year={2024}
}

@article{zhang2024bidirectional,
  title={Bidirectional Corrective Model-Contrastive Federated Adversarial Training},
  author={Zhang, Yuyue and Shi, Yicong and Zhao, Xiaoli},
  journal={Electronics},
  volume={13},
  number={18},
  pages={3745},
  year={2024},
  publisher={MDPI}
}

@inproceedings{zhang2024federated,
  title={When federated recommendation meets cold-start problem: Separating item attributes and user interactions},
  author={Zhang, Chunxu and Long, Guodong and Zhou, Tianyi and Zhang, Zijian and Yan, Peng and Yang, Bo},
  booktitle={Proceedings of the ACM Web Conference 2024},
  pages={3632--3642},
  year={2024}
}

@article{wang2024horizontal,
  title={Horizontal federated recommender system: A survey},
  author={Wang, Lingyun and Zhou, Hanlin and Bao, Yinwei and Yan, Xiaoran and Shen, Guojiang and Kong, Xiangjie},
  journal={ACM Computing Surveys},
  volume={56},
  number={9},
  pages={1--42},
  year={2024},
  publisher={ACM New York, NY}
}

@article{andriushchenko2020understanding,
  title={Understanding and improving fast adversarial training},
  author={Andriushchenko, Maksym and Flammarion, Nicolas},
  journal={Advances in Neural Information Processing Systems},
  volume={33},
  pages={16048--16059},
  year={2020}
}

@article{balaji2019instance,
  title={Instance adaptive adversarial training: Improved accuracy tradeoffs in neural nets},
  author={Balaji, Yogesh and Goldstein, Tom and Hoffman, Judy},
  journal={arXiv preprint arXiv:1910.08051},
  year={2019}
}

@article{yang2024understanding,
  title={Understanding server-assisted federated learning in the presence of incomplete client participation},
  author={Yang, Haibo and Qiu, Peiwen and Khanduri, Prashant and Fang, Minghong and Liu, Jia},
  journal={arXiv preprint arXiv:2405.02745},
  year={2024}
}

@article{elsayed2018large,
  title={Large margin deep networks for classification},
  author={Elsayed, Gamaleldin and Krishnan, Dilip and Mobahi, Hossein and Regan, Kevin and Bengio, Samy},
  journal={Advances in neural information processing systems},
  volume={31},
  year={2018}
}

@article{chen2024post,
  title={Post-training attribute unlearning in recommender systems},
  author={Chen, Chaochao and Zhang, Yizhao and Li, Yuyuan and Wang, Jun and Qi, Lianyong and Xu, Xiaolong and Zheng, Xiaolin and Yin, Jianwei},
  journal={ACM Transactions on Information Systems},
  volume={43},
  number={1},
  pages={1--28},
  year={2024},
  publisher={ACM New York, NY, USA}
}

@inproceedings{feng2025plug,
  title={Plug and Play: Enabling Pluggable Attribute Unlearning in Recommender Systems},
  author={Feng, Xiaohua and Li, Yuyuan and Yu, Fengyuan and Zhang, Li and Chen, Chaochao and Zheng, Xiaolin},
  booktitle={Proceedings of the ACM on Web Conference 2025},
  pages={2689--2699},
  year={2025}
}

@inproceedings{yu2025lego,
  title={LEGO: A Lightweight and Efficient Multiple-Attribute Unlearning Framework for Recommender Systems},
  author={Yu, Fengyuan and Li, Yuyuan and Feng, Xiaohua and Fang, Junjie and Wang, Tao and Chen, Chaochao},
  booktitle={Proceedings of the 33rd ACM International Conference on Multimedia},
  pages={6242--6251},
  year={2025}
}

@article{li2024survey,
  title={A survey on recommendation unlearning: Fundamentals, taxonomy, evaluation, and open questions},
  author={Li, Yuyuan and Feng, Xiaohua and Chen, Chaochao and Yang, Qiang},
  journal={arXiv preprint arXiv:2412.12836},
  year={2024}
}

@inproceedings{liu2024fedbcgd,
  title={Fedbcgd: Communication-efficient accelerated block coordinate gradient descent for federated learning},
  author={Liu, Junkang and Shang, Fanhua and Liu, Yuanyuan and Liu, Hongying and Li, Yuangang and Gong, YunXiang},
 booktitle={Proceedings of the 32nd ACM International Conference on Multimedia},
  pages={2955--2963},
  year={2024}
}

@article{zhou2023fastpillars,
  title={FastPillars: A Deployment-friendly Pillar-based 3D Detector},
  author={Zhou, Sifan and Tian, Zhi and Chu, Xiangxiang and Zhang, Xinyu and Zhang, Bo and Lu, Xiaobo and Feng, Chengjian and Jie, Zequn and Chiang, Patrick Yin and Ma, Lin},
  journal={arXiv preprint arXiv:2302.02367},
  year={2023}
}

@inproceedings{pillarhist,
  title={Pillarhist: A quantization-aware pillar feature encoder based on height-aware histogram},
  author={Zhou, Sifan and Yuan, Zhihang and Yang, Dawei and Hu, Xing and Qian, Jian and Zhao, Ziyu},
  booktitle={Proceedings of the Computer Vision and Pattern Recognition Conference},
  pages={27336--27345},
  year={2025}
}

@inproceedings{wang2025robin,
  title={ROBIN: A Novel Framework for Accelerating Robust Multi-Variant Training},
  author={Wang, Yan and Wang, Xingbin and Su, Yulan and Zhang, Sisi and Lin, Zechao and Meng, Dan and Hou, Rui},
  booktitle={Proceedings of the 30th Asia and South Pacific Design Automation Conference},
  pages={1120--1125},
  year={2025}
}

@inproceedings{ENCODER,
  title={Encoder: Entity mining and modification relation binding for composed image retrieval},
  author={Li, Zixu and Chen, Zhiwei and Wen, Haokun and Fu, Zhiheng and Hu, Yupeng and Guan, Weili},
 booktitle={Proceedings of the AAAI Conference on Artificial Intelligence},
  volume={39},
  number={5},
  pages={5101--5109},
  year={2025}
}

@article{FineCIR,
  title={FineCIR: Explicit Parsing of Fine-Grained Modification Semantics for Composed Image Retrieval},
  author={Li, Zixu and Fu, Zhiheng and Hu, Yupeng and Chen, Zhiwei and Wen, Haokun and Nie, Liqiang},
  journal={https://arxiv.org/abs/2503.21309},
  year={2025}
}

@inproceedings{OFFSET, 
  title = {OFFSET: Segmentation-based Focus Shift Revision for Composed Image Retrieval}, 
  author = {Chen, Zhiwei and Hu, Yupeng and Li, Zixu and Fu, Zhiheng and Song, Xuemeng and Nie, Liqiang}, 
  booktitle = {Proceedings of the ACM International Conference on Multimedia}, 
  pages = {6113–6122}, 
  year = {2025}
}

@inproceedings{lu2025dmmd4sr,
  title={DMMD4SR: Diffusion Model-based Multi-level Multimodal Denoising for Sequential Recommendation},
  author={Lu, Weihai and Yin, Li},
  booktitle={Proceedings of the 33rd ACM International Conference on Multimedia},
  pages={6363--6372},
  year={2025}
}

@inproceedings{lu2025dammfnd,  
  title={DAMMFND: Domain-Aware Multimodal Multi-view Fake News Detection},
  author={Lu, Weihai and Tong, Yu and Ye, Zhiqiu},
  booktitle={Proceedings of the AAAI Conference on Artificial Intelligence},
  volume={39},
  number={1},
  pages={559--567},
  year={2025}
}

@inproceedings{glorot2010understanding,
  title={Understanding the difficulty of training deep feedforward neural networks},
  author={Glorot, Xavier and Bengio, Yoshua},
  booktitle={Proceedings of the thirteenth international conference on artificial intelligence and statistics},
  pages={249--256},
  year={2010},
  organization={JMLR Workshop and Conference Proceedings}
}

@inproceedings{liu2024qoe,
  title={QoE-Aware Online Auction Mechanism for UAV-enabled Crowd-sensing},
  author={Liu, Ying and Cai, Bohan and Zhi, Jiawang and Wu, Gang and Xia, Xiaoyu},
  booktitle={2024 IEEE International Conference on Web Services (ICWS)},
  pages={654--664},
  year={2024},
  organization={IEEE}
}

@inproceedings{liu2025fine,
  title={Fine-Grained Open-Vocabulary Object Detection with Fined-Grained Prompts: Task, Dataset and Benchmark},
  author={Liu, Ying and Hua, Yijing and Chai, Haojiang and Wang, Yanbo and Ye, TengQi},
  booktitle={2025 IEEE International Conference on Robotics and Automation (ICRA)},
  pages={13860--13867},
  year={2025},
  organization={IEEE}
}

@article{zhang2025fedfact,
  title={FedFACT: A Provable Framework for Controllable Group-Fairness Calibration in Federated Learning},
  author={Zhang, Li and Han, Zhongxuan and Feng, Xiaohua and Zhang, Jiaming and Li, Yuyuan and others},
  journal={Advances in Neural Information Processing Systems},
  volume={38},
  year={2025}
}

@article{chen2024integration,
  title={Integration of large language models and federated learning},
  author={Chen, Chaochao and Feng, Xiaohua and Li, Yuyuan and Lyu, Lingjuan and Zhou, Jun and Zheng, Xiaolin and Yin, Jianwei},
  journal={Patterns},
  volume={5},
  number={12},
  year={2024},
  publisher={Elsevier}
}

@inproceedings{feng2024fine,
  title={Fine-grained pluggable gradient ascent for knowledge unlearning in language models},
  author={Feng, XiaoHua and Chen, Chaochao and Li, Yuyuan and Lin, Zibin},
  booktitle={Proceedings of the 2024 Conference on Empirical Methods in Natural Language Processing},
  pages={10141--10155},
  year={2024}
}

@inproceedings{feng2025controllable,
  title={Controllable Unlearning for Image-to-Image Generative Models via Constrained Optimization},
  author={Feng, XiaoHua and Li, Yuyuan and Chen, Chaochao and Zhang, Li and Li, Longfei and ZHOU, JUN and Zheng, Xiaolin},
  booktitle={The Thirteenth International Conference on Learning Representations},
  year={2025}
}

@article{feng2025survey,
  title={A survey on generative model unlearning: Fundamentals, taxonomy, evaluation, and future direction},
  author={Feng, Xiaohua and Zhang, Jiaming and Yu, Fengyuan and Wang, Chengye and Zhang, Li and Li, Kaixiang and Li, Yuyuan and Chen, Chaochao and Yin, Jianwei},
  journal={arXiv preprint arXiv:2507.19894},
  year={2025}
}

\end{document}